%% file: main.tex
\documentclass[a4paper,UKenglish,cleveref, autoref, thm-restate]{lipics-v2021}

\usepackage{amsmath}
\usepackage{amssymb}
\usepackage{amsfonts}
\usepackage{amsthm} 
\usepackage{graphicx}
\usepackage{comment}

\usepackage{thm-restate}

\usepackage{hyperref}
\usepackage{here}

\usepackage{cite}
\usepackage{enumerate}

\usepackage[linesnumbered,noend,ruled,vlined]{algorithm2e}

\SetKwInput{KwInput}{Input}     
\SetKwInput{KwOutput}{Output}   
\SetKw{KwTo}{to}

\usepackage[framemethod=tikz]{mdframed}

\global\long\def\M{\mathcal{M}}%
\global\long\def\I{\mathcal{I}}%
\global\long\def\tO{\tilde{O}}%

\global\long\def\eps{\varepsilon}%
\global\long\def\epsilon{\varepsilon}%

\global\long\def\TI{\mathcal{T}_\textrm{ind}}%
\global\long\def\TR{\mathcal{T}_\textrm{rank}}%

\usepackage{tabularx}
\usepackage{environ}

\usepackage{cite}
\usepackage{url}

\usepackage[normalem]{ulem}
\usepackage{color}

\title{Deterministic $(2/3 - \varepsilon)$-Approximation of Matroid Intersection Using Nearly-Linear Independence-Oracle Queries
}
\titlerunning{Deterministic $(2/3 - \varepsilon)$-Approximation of Matroid Intersection Using Nearly-Linear Queries}

\author{Tatsuya Terao}{Research Institute for Mathematical Sciences, Kyoto University}{ttatsuya@kurims.kyoto-u.ac.jp}{https://orcid.org/0000-0002-3530-2194}{}

\authorrunning{T.\ Terao}

\Copyright{Tatsuya Terao}


\acknowledgements{
The author thanks Yusuke Kobayashi for his generous support and helpful comments on the manuscript.
This work was partially supported by the joint project of Kyoto University and Toyota Motor Corporation, titled ``Advanced Mathematical Science for Mobility Society'' and by JSPS KAKENHI Grant Number JP24KJ1494.
}


\date{}

\hideLIPIcs
\nolinenumbers

\begin{document}

\maketitle

\begin{abstract}

In the matroid intersection problem, we are given two matroids $\mathcal{M}_1 = (V, \mathcal{I}_1)$ and $\mathcal{M}_2 = (V, \mathcal{I}_2)$ defined on the same ground set $V$ of $n$ elements, and the objective is to find a common independent set $S \in \mathcal{I}_1 \cap \mathcal{I}_2$ of largest possible cardinality, denoted by $r$. In this paper, we consider a deterministic matroid intersection algorithm with only a nearly linear number of independence oracle queries. Our contribution is to present a deterministic $O(\frac{n}{\varepsilon} + r \log r)$-independence-query $(2/3-\varepsilon)$-approximation algorithm for any $\varepsilon > 0$. Our idea is very simple: we apply a recent $\tilde{O}(n \sqrt{r}/\varepsilon)$-independence-query $(1 - \varepsilon)$-approximation algorithm of Blikstad [ICALP 2021], but terminate it before completion. Moreover, we also present a semi-streaming algorithm for $(2/3 -\varepsilon)$-approximation of matroid intersection in $O(1/\varepsilon)$ passes.

\end{abstract}

\input{introduction.tex}

\input{preliminaries.tex}

\input{approx_algo.tex}

\input{approx_algo_using_rank_queries.tex}

\input{streaming_algo.tex}

\input{conclusion.tex}

\bibliography{bib2doi}

\end{document}

%% file: introduction.tex
\section{Introduction} \label{sec:introduction}

The \emph{matroid intersection problem} is one of the most fundamental problems in combinatorial optimization.
In this problem, we are given two matroids $\M_1 = (V, \mathcal{I}_1)$ and $\M_2 = (V, \mathcal{I}_2)$ defined on the same ground set $V$ of $n$ elements, and the objective is to find a common independent set $S \in \mathcal{I}_1 \cap \mathcal{I}_2$ of largest possible cardinality, denoted by $r$.
This problem generalizes many important combinatorial optimization problems such as bipartite matching, packing spanning trees, and arborescences in directed graphs. 
Furthermore, it has also several applications outside of traditional combinatorial optimization such as electrical engineering \cite{murota2010matrices, recski2013matroid}.

To design an algorithm for arbitrary matroids, it is common to consider an oracle model: an algorithm accesses a matroid through an oracle.
The most standard and well-studied oracle is an {\em independence oracle}, which takes as input a subset $S \subseteq V$ and outputs whether $S$ is independent or not.
Many studies consider the following research question: how few queries can we solve the matroid intersection problem with?

Starting the work of Edmonds \cite{edmonds1970submodular, edmonds1979matroid}, many algorithms with polynomial query complexity for the matroid intersection problem have been
studied \cite{aigner1971matching, lawler1975matroid, cunningham1986improved, shigeno1995dual, lee2015faster, chekuri2016fast, huang2016exact, nguyen2019note, chakrabarty2019faster, blikstad2021breaking_STOC, blikstad2021breaking, u2022subquadratic, blikstad2023fast, berczi2023matroid,quanrud:LIPIcs.ICALP.2024.118,blikstad2024efficient}.
Edmonds \cite{edmonds1970submodular} developed the first polynomial-query algorithm by reduction to the matroid partitioning problem. This algorithm requires $O(n^4)$ independence oracle queries.
More direct algorithms were given by Aigner--Dowling \cite{aigner1971matching} and Lawler \cite{lawler1975matroid}.
The algorithm of Lawler requires $O(nr^2)$ independence oracle queries.
In 1986, Cunningham \cite{cunningham1986improved} presented an $O(nr^{3/2})$-independence-query algorithm by using a blocking flow approach, which is akin to the bipartite matching algorithm of Hopcroft--Karp \cite{hopcroft1973n}.
Chekuri--Quanrud \cite{chekuri2016fast} pointed out that, for any $\eps > 0$, an $O(nr/\eps)$-independence-query $(1 - \eps)$-approximation algorithm can be obtained by terminating Cunningham's algorithm early.
This idea comes from the well-known fact that, for the bipartite matching problem, a linear-time $(1 - \eps)$-approximation algorithm can be obtained by terminating Hopcroft--Karp's algorithm early.
Recently, Nguy$\tilde{{\hat{\text{e}}}}$n~\cite{nguyen2019note} and Chakrabarty--Lee--Sidford--Singla--Wong \cite{chakrabarty2019faster} independently presented a new binary search technique that can efficiently find edges in the exchange graph and developed a combinatorial $\tilde{O}(n r)$-independence-query exact algorithm.\footnote{The $\tO$ notation omits factors polynomial in $\log n$.}
Chakrabarty et al.~also presented a new augmenting sets technique and developed an $\tO(n^{1.5}/\epsilon^{1.5})$-independence-query $(1 - \epsilon)$-approximation algorithm.
The techniques developed by Nguy$\tilde{{\hat{\text{e}}}}$n and Chakrabarty et al.~have been used in recent several studies on fast algorithms for other matroid problems \cite{blikstad:LIPIcs.ICALP.2022.25, terao:LIPIcs.ICALP.2023.104, viswanathan2023general, quanrud2024faster, kobayashi2024subqudratic, buchbinder2024deterministic}.
Blikstad--van den Brand-Mukhopadhyay--Nanongkai~\cite{blikstad2021breaking_STOC} first broke the $\tilde{O}(n^2)$-independence-query bound for exact matroid intersection algorithms.
Blikstad \cite{blikstad2021breaking} improved the independence query complexity of a $(1 - \epsilon)$-approximation algorithm to $\tilde{O}(n \sqrt{r} / \eps)$.
Blikstad's improvement on the $(1 - \epsilon)$-approximation algorithm resulted in a randomized $\tilde{O}(nr^{3/4})$-independence-query exact algorithm and a deterministic $\tilde{O}(nr^{5/6})$-independence-query exact algorithm.
Recently, Quanrud \cite{quanrud:LIPIcs.ICALP.2024.118} presented a randomized $\tilde{O}_{\eps}(n + r^{3/2})$-independence-query $(1 - \epsilon)$-approximation algorithm\footnote{The $\tilde{O}_{\eps}$ notation omits factors polynomial in $\eps$ and $\log n$.}.
Remarkably, Blikstad-Tu \cite{blikstad2024efficient} very recently presented a randomized $\tilde{O}_{\eps}(n)$-independence-query $(1 - \epsilon)$-approximation algorithm.\footnote{
The results in this manuscript were obtained independently of the remarkable result by Blikstad--Tu \cite{blikstad2024efficient}.
}



Both the recent $(1 - \epsilon)$-approximation algorithms by Quanrud \cite{quanrud:LIPIcs.ICALP.2024.118} and Blikstad--Tu \cite{blikstad2024efficient} are randomized.
Therefore, it remains an open question whether a deterministic $\tilde{O}_{\eps}(n)$ independence-query $(1 - \epsilon)$-approximation algorithm can be achieved for almost the entire range of $r$.\footnote{Blikstad--Tu \cite{blikstad2024efficient} presented a deterministic $(1 - \epsilon)$-approximation algorithm using a strictly linear number of independence oracle queries only when $r = \Theta(n)$ (and $\eps$ is constant).}
Then, we consider a deterministic matroid intersection algorithm with only a nearly linear number of independence oracle queries.
It is well-known that any maximal common independent set is at least half the size of a maximum common independent set (see e.g., \cite[Proposition 13.26]{korte2011combinatorial}).
Thus, the natural greedy algorithm---pick an element whenever possible---can be regarded as a deterministic $1/2$-approximation algorithm using a linear number of independence oracle queries.
In fact, even beating the trivial $1/2$-approximation ratio for deterministic nearly-linear-independence-query algorithms remains an open problem for almost the entire range of $r$.\footnotemark[4]



    




In this work, we present a deterministic $(2/3 - \eps)$-approximation algorithm using a nearly linear number of independence oracle queries.

\begin{restatable}{theorem}{indalgo} \label{thm:main_matroid_intersection_ind}
    Given two matroids $\M_1 = (V, \I_1)$ and $\M_2 = (V, \I_2)$ on the same ground set $V$, for any $\eps > 0$,
    there is a deterministic algorithm that finds a common independent set $S \in \I_1 \cap \I_2$ with $|S| \geq (2/3 - \varepsilon) r$, using $O(\frac{n}{\varepsilon} + r \log r)$ independence oracle queries.
\end{restatable}

Our algorithm uses only a strictly linear number of independence oracle queries when $r = O(n / \log n)$ (and $\eps$ is constant).
Prior to our work, Guruganesh--Singla \cite{guruganesh2017online} presented a randomized $O(n)$-independence-query algorithm that achieves $(1/2 + \delta)$-approximation for some small constant $\delta > 0$, which was the best approximation ratio for (possibly randomized) algorithms using only a strictly linear number of independence oracle queries.\footnotemark[4]

We note that the time complexity of our algorithm is dominated by the independence oracle queries.
That is, our algorithm in Theorem \ref{thm:main_matroid_intersection_ind} has time complexity $O\left((\frac{n}{\varepsilon} + r \log r)\cdot \TI\right)$, where $\TI$ denotes the maximum running time of a single independence oracle query.

Our idea is very simple: we apply a recent $\tilde{O}(n \sqrt{r}/\varepsilon)$-independence-query $(1 - \varepsilon)$-approximation algorithm of Blikstad \cite{blikstad2021breaking}, but terminate it before completion.
We observe that a $(2/3 - \eps)$-approximate solution can be obtained by terminating Blikstad's algorithm early.
As such, our algorithm does not introduce technical novelty, but we believe that it enhances the understanding of algorithms for the matroid intersection problem.

We also implement our algorithm using a {\em rank oracle}, which takes as input a subset $S \subseteq V$ and outputs the size of the maximum cardinality independent subset of $S$.
Several recent studies have considered fast matroid intersection algorithms in the rank oracle model \cite{lee2015faster, chakrabarty2019faster, u2022subquadratic}.
We note that the rank oracle is more powerful than the independence oracle, since a single query to the rank oracle can determine whether a given set is independent or not.

Chakrabarty--Lee--Sidford--Singla--Wong \cite{chakrabarty2019faster} presented an $\tilde{O}(n \sqrt{r})$-rank-query exact algorithm and an $O(\frac{n}{\varepsilon} \log n)$-rank-query $(1 - \varepsilon)$-approximation algorithm.
Their $(1 - \varepsilon)$-approximation algorithm requires nearly linear rank oracle queries.
Our second result is to obtain a $(2/3 - \eps)$-approximation algorithm using a strictly linear number of rank oracle queries.

\begin{restatable}{theorem}{rankalgo} \label{thm:main_matroid_intersection_rank}
    Given two matroids $\M_1 = (V, \I_1)$ and $\M_2 = (V, \I_2)$ on the same ground set $V$, for any $\eps > 0$,
    there is a deterministic algorithm that finds a common independent set $S \in \I_1 \cap \I_2$ with $|S| \geq (2/3 - \varepsilon) r$, using $O(\frac{n}{\varepsilon})$ rank oracle queries.
\end{restatable}

We note that the time complexity of our algorithm is dominated by the rank oracle queries.
That is, our algorithm in Theorem \ref{thm:main_matroid_intersection_rank} has time complexity $O\left(\frac{n}{\varepsilon} \cdot \TR \right)$, where $\TR$ denotes the maximum running time of a single rank oracle query.

Surprisingly, our algorithm can also be implemented in the semi-streaming model of computation. 
The streaming model for graph problems was initiated by Feigenbaum--Kannan--McGregor--Suri--Zhang \cite{feigenbaum2005graph}, which is called the semi-streaming model.
In this model, an algorithm computes over a graph stream using a limited space of $\tilde{O}(n)$ bits, where $n$ is the number of vertices.
The maximum matching problem is one of the most studied problems in the graph streaming setting \cite{feigenbaum2005graph,mcgregor2005finding,eggert2009bipartite,ahn2011laminar,goel2012communication,eggert2012bipartite,konrad2012maximum,ahn2013linear,kapralov2013better,esfandiari2016finding,kale_et_al:LIPIcs.APPROX-RANDOM.2017.15,ahn2018access,konrad2018simple,tirodkar2018deterministic,gamlath2019weighted,assadi2021auction,assadi_et_al:LIPIcs.ICALP.2021.19,assadi2022semi,fischer2022deterministic,huang20231, bernstein2023improved,konrad_et_al:LIPIcs.STACS.2023.41, feldman2024maximum, assadi2024simple}.
The first of these studies was a $(2/3 - \eps)$-approximation algorithm for the bipartite maximum matching problem by Feigenbaum et al.~\cite{feigenbaum2005graph}, which uses $\tilde{O}(n)$ bits of space and $O(\log (1/\eps) / \eps)$ passes.
In particular, whether a $(1 - \epsilon)$ approximation can be achieved in $\text{poly}(1/\eps)$ passes in the semi-streaming model for general graphs was a significant open problem until it was solved by Fischer--Mitrovi{\'c}--Uitto~\cite{fischer2022deterministic}.

Since the matroid intersection problem is a well-known generalization of the bipartite matching problem, there are also several studies on the matroid intersection problem (and its generalizations) in the streaming setting in the literature \cite{chekuri2015streaming, feldman2018less, garg2023semi, huang2024robust, huang_et_al:LIPIcs.APPROX/RANDOM.2020.62, huang2023fpt, chakrabarti2015submodular}.
For the matroid intersection problem, in the semi-streaming model, we consider algorithms such that the memory usage is $O((r_1 + r_2) \textrm{polylog}(r_1 + r_2))$, where $r_1$ and $r_2$ denote the ranks of matroids $\M_1$ and $\M_2$, respectively.
We note that this memory requirement is natural when we formulate the bipartite matching problem as the matroid intersection problem.
In this work, we focus on constant-passes ($\text{poly}(1/\eps)$-passes) semi-streaming algorithms for the matroid intersection problem, particularly given the significance of constant-passes semi-streaming algorithms for the matching problem.\footnote{Based on the results from Assadi {\cite{assadi2024simple}} and Quanrud {\cite{quanrud:LIPIcs.ICALP.2024.118}}, we believe that it is able to develop a semi-streaming $(1 - \eps)$-approximation algorithm with similar space usage. However, such an algorithm would require $O(\log (n) \cdot \text{poly}(1/\eps))$ passes. In the context of research on streaming matching algorithms, the distinctions between $O(\text{poly}(1/\eps))$ passes and $O(\log (n) \cdot \text{poly}(1/\eps))$ are well recognized and crucial; see {\cite[Table 1 in the arXiv version]{assadi2024simple}}.
We also believe that the result by Blikstad--Tu \cite{blikstad2024efficient} is unlikely to immediately lead to constant-passes semi-streaming algorithm for the matroid intersection problem.
This is because, while Blikstad--Tu used some idea from the constant-passes semi-streaming algorithm of Assadi-Liu-Tarjan \cite{assadi2021auction} for the bipartite matching problem, Blikstad--Tu relies on the technique from Quanrud {\cite{quanrud:LIPIcs.ICALP.2024.118}}.
Quanrud's algorithm used some idea from the $O(\log (n) \cdot \text{poly}(1/\eps))$-passes semi-streaming algorithm by Assadi \cite{assadi2024simple}.
}
Recently, Huang--Sellier \cite{huang2024robust} presented a $(2/3 - \eps)$-approximation 
semi-streaming algorithm for the matroid intersection problem in the random-order streaming model.\footnote{We note that their algorithm requires too many queries to compute the density-based decomposition. Thus, their algorithm does not imply a nearly-linear-independence-query $(2/3 - \eps)$-approximation algorithm.}

Our third result is to obtain a $(2/3 - \eps)$-approximation semi-streaming algorithm with $O(\frac{1}{\varepsilon})$ passes in the adversary-order streaming model.

\begin{restatable}{theorem}{streamingalgo} \label{thm:main_matroid_intersection_streaming}
    Given two matroids $\M_1 = (V, \I_1)$ and $\M_2 = (V, \I_2)$ on the same ground set $V$, for any $\eps > 0$,
    there is a deterministic semi-streaming algorithm that finds a common independent set $S \in \I_1 \cap \I_2$ with $|S| \geq (2/3 - \varepsilon) r$, using $O(\frac{r_1 + r_2}{\varepsilon})$ memory and $O(\frac{1}{\varepsilon})$ passes.
\end{restatable}

Our algorithm in Theorem \ref{thm:main_matroid_intersection_streaming} is a generalization of the bipartite matching semi-streaming algorithm of Feigenbaum et al.~\cite{feigenbaum2005graph}.



%% file: preliminaries.tex
\section{Preliminaries} \label{sec:preliminaries}

\subsection{Notation} \label{subsec:matroids}

Here, we provide the basic notations and conventions used throughout the paper.

\paragraph*{Set Notation.} \label{subsec:basic_notation}

For a set $A$ and an element $x$, we will often write $A + x := A \cup \{ x \}$ and $A - x := A \setminus \{ x \}$.
For two sets $A$ and $B$, we will also write $A + B := A \cup B$ and $A - B := A \setminus B$ when co confusion can arise.

\paragraph*{Matroid.}
    A pair $\M = (V, \mathcal{I})$ of a finite set $V$ and a non-empty set family $\mathcal{I} \subseteq 2^{V}$ is called a {\em matroid} if the following properties are satisfied.
    
    \begin{description}
        \item[(Downward closure property)] 
        If $S \in \mathcal{I}$ and $S' \subseteq S$, then $S' \in \mathcal{I}$.
        \item[(Augmentation property)]
        If $S, S' \in \mathcal{I}$ and $|S'| < |S|$, then there exists $v \in S \setminus S'$ such that $S' + v \in \mathcal{I}$.
    \end{description}
A set $S \subseteq V$ is called {\em independent} if $S \in \mathcal{I}$ and {\em dependent} otherwise.

\paragraph*{Rank.}
    For a matroid $\M = (V, \mathcal{I})$, we define the {\em rank} of $\M$ as ${\rm rank}(\M) = \max \{ |S| \mid S \in \mathcal{I} \}$.
    In addition, for any $S \subseteq V$, we define the {\em rank} of $S$ as ${\rm rank}_{\M}(S) = \max \{ |T| \mid T \subseteq S, T \in \mathcal{I} \}$.

\paragraph*{Matroid Intersection.}
    Given two matroids $\mathcal{M}_1 = (V, \mathcal{I}_1)$ and $\mathcal{M}_2 = (V, \mathcal{I}_2)$ defined on the same ground set $V$, a {\em common independent set} $S$ is a set in $\I_1 \cap \I_2$.
    In the \emph{matroid intersection problem}, the aim is to find a largest common independent set, whose cardinality we denote by $r$.

\subsection{Exchange Graph and Augmenting Path} \label{subsec:augmenting_path}

We provide the standard definitions of {\em exchange graph} and {\em augmenting paths} and known lemmas used in recent fast matroid intersection algorithms.
Many matroid intersection algorithms use an approach of iteratively finding augmenting paths in the exchange graph.

\begin{definition}[Exchange Graph]
    Consider a common independent set $S \in \mathcal{I}_1 \cap \mathcal{I}_2$ of two matroids $\M_1$ and $\M_2$.
    The {\em exchange graph} is defined as a directed graph $G(S) = (V \cup \{s, t \}, E)$, with $s, t \notin V$ and $E = E' \cup E'' \cup E_s \cup E_t$, where
    \begin{align*}
        E' = & \{(u, v) \mid u \in S, v \in V \setminus S, S - u + v \in \mathcal{I}_1 \} ,\\
        E'' = & \{(v, u) \mid u \in S, v \in V \setminus S, S - u + v \in \mathcal{I}_2 \} ,\\
        E_s = & \{(s, v) \mid v \in V \setminus S, S + v \in \mathcal{I}_1 \} , \text{and} \\
        E_t = & \{(v, t) \mid v \in V \setminus S, S + v \in \mathcal{I}_2 \} .
    \end{align*}
\end{definition}

\begin{lemma}[Shortest Augmenting Path; see {\cite[Theorem 41.2]{schrijver2003combinatorial}}]
    Let $s, v_1, v_2, \dots , v_{\ell - 1}, t$ be a shortest $(s, t)$-path in the exchange graph $G(S)$ relative to a common independent set $S \in \I_1 \cap \I_2$.
    Then, $S' = S + v_1 - v_2 + \cdots - v_{\ell - 2} + v_{\ell - 1} \in \I_1 \cap \I_2$.
\end{lemma}



Cunningham's \cite{cunningham1986improved} matroid intersection algorithm and recent fast matroid intersection algorithms \cite{chekuri2016fast, nguyen2019note, chakrabarty2019faster, blikstad2021breaking_STOC, blikstad2021breaking} rely on the following lemma.

\begin{lemma}[from {\cite[Corollary 2.2]{cunningham1986improved}}]\label{lem:total_bound_matroid_intersection}
    Let $S$ be a common independent set which is not maximum.
    Then, there exists an $(s, t)$-path of length at most $\frac{2 |S|}{r - |S|} + 2$ in the exchange graph $G(S)$.
\end{lemma}


Chakrabarty--Lee--Sidford--Singla--Wong \cite{chakrabarty2019faster} presented a new binary search technique that can efficiently find edges in the exchange graph. (This technique was developed independently by Nguy$\tilde{{\hat{\text{e}}}}$n~\cite{nguyen2019note}.)

\begin{lemma}[Binary Search Technique from \cite{nguyen2019note,chakrabarty2019faster}] \label{lem:binary_search_ind}
    Given a matroid $\M = (V, \mathcal{I})$, an independent set $S \in \mathcal{I}$, an element $v \in V \setminus S$, and $T \subseteq S$, using $O(\log |T|)$ independence oracle queries, we can find an element $u \in T$ such that $S + v - u \in \mathcal{I}$ or otherwise determine that no such element exists.
    Furthermore, if there exists no such an element, then this procedure uses only one independence oracle query.\footnote{This is because there exists such an element if and only if $S + v - T \in \mathcal{I}$.}
\end{lemma}

Let \texttt{FindExchange}$(\M, S, v, T)$ be the procedure that implements Lemma \ref{lem:binary_search_ind}.

Let $D_i$ be the set of elements $v \in V$ such that the distance from $s$ to $v$ is exactly $i$ in the exchange graph $G(S)$.
Chakrabarty et al.~\cite{chakrabarty2019faster} showed that the distance layers $D_i$ can be computed efficiently by using the binary search technique.
In their algorithm, they compute odd and even layers separately because \texttt{FindExchange} cannot find both incoming and outgoing edges due to restrictions on the allowed queries; see \cite[Algorithm 6 and Lemma 19]{chakrabarty2019faster}.
Their algorithm requires $O(n r + r^2 \log r)$ independence oracle queries to compute all distance layers $D_1, D_2, \ldots , D_{2r+1}$; see \cite[Remark just after Proof of Theorem 18]{chakrabarty2019faster}.
In our matroid intersection algorithm, we only need to compute $D_1$, $D_2$, and $D_3$.
Their algorithm can compute $D_1$, $D_2$, and $D_3$ using $O(n + r \log r)$ independence oracle queries.
See \texttt{GetDistance} (Algorithm \ref{alg:GetDistance}) for the pseudocode of the algorithm.

\begin{algorithm}[t]
    \KwInput{a common independent set $S \in \I_1 \cap \I_2$ }
    \KwOutput{distance layers $D_1 \subseteq V \setminus S, D_2 \subseteq S$, and $D_3 \subseteq V \setminus S$}
    $D_1 \gets \emptyset$, $D_2 \gets \emptyset$, $D_3 \gets \emptyset$ \\
    \For{$v \in V \setminus S$ with $S + v \in \I_1$} { \label{line:bfs1}
        $D_1 \gets D_1 + v$
    }
    $Q \gets D_1$, $T \gets S$ \\
    \While{$Q$ contains some $v$} {
        \While{$u = $\texttt{\textup{ FindExchange}}$(\M_2, S, v, T)$ satisfies $u \neq \emptyset$} { \label{line:useoffindexchange}
            $D_2 \gets D_2 + u$, $T \gets T - u$ \\
        }
        $Q \gets Q - v$ \\
    }
    \For{$v \in V \setminus (S \cup D_1)$ with $S + v - D_2 \in \I_1$} { \label{line:bfs2}
        $D_3 \gets D_3 + v$
    }
    \caption{\texttt{GetDistance}}\label{alg:GetDistance}
\end{algorithm}

\begin{lemma}[follows from {\cite[Lemma 19]{chakrabarty2019faster}}] \label{lem:find_dist_layers_ind}
    Given a common independent set $S \in \I_1 \cap \I_2$ of two matroids $\M_1$ and $\M_2$, using $O(n + r \log r)$ independence oracle queries, we can find the distance layers $D_1 \subseteq V \setminus S, D_2 \subseteq S$, and $D_3 \subseteq V \setminus S$.
\end{lemma}

For completeness, we give a proof of Lemma \ref{lem:find_dist_layers_ind}.

\begin{proof}
    The procedure \texttt{GetDistance} (Algorithm \ref{alg:GetDistance}) simply performs a breadth first search in the exchange graph $G(S)$.
    Thus, the procedure \texttt{GetDistance} correctly computes the distance layers $D_1$, $D_2$, and $D_3$.
    Here, each element $u \in D_2$ is found by \texttt{FindExchange} only once.
    Thus, the number of \texttt{FindExchange} calls that do not output $\emptyset$ is $|D_2| \leq r$, and the number of \texttt{FindExchange} calls that output $\emptyset$ is $|D_1| \leq n$.
    Hence, by Lemma \ref{lem:binary_search_ind}, the number of independence oracle queries used in Line \ref{line:useoffindexchange} is $O(n + r \log r)$.
    In addition, the number of independence oracle queries used in Lines \ref{line:bfs1} and \ref{line:bfs2} is $O(n)$, which completes the proof.
\end{proof}

\section{Augmenting Sets Technique} \label{subsec:augmenting_sets}

In this section, we recall the definition and properties of {\em augmenting sets}, which were introduced as a generalization of {\em augmenting paths} by Chakrabarty--Lee--Sidford--Singla--Wong \cite{chakrabarty2019faster}.
The augmenting sets technique plays a crucial role in the $(1 - \eps)$-approximation algorithms of Chakrabarty et al.~\cite{chakrabarty2019faster} and Blikstad \cite{blikstad2021breaking}.

In our matroid intersection algorithm, we only need to find an augmenting set in the exchange graph whose shortest $(s, t)$-path length is $4$.
Thus, we introduce the definition and properties of augmenting sets when restricted to the case where the length of a shortest augmenting path is $4$.

\begin{definition}[Augmenting Sets from {\cite[Definition 24]{chakrabarty2019faster}}]
    Let $S \in \I_1 \cap \I_2$ be such that shortest augmenting paths in $G(S)$ have length $4$.
    We say that a collection of sets $\Pi := (B_1, A_1, B_2)$ form an \emph{augmenting set} (of \emph{width} $w$) in $G(S)$ if the following conditions are satisfied:
    
    \begin{enumerate}[(a)]
        \item $B_1 \subseteq D_{1}$, $A_1 \subseteq D_{2}$, and $B_2 \subseteq D_{3}$.
        \item $|B_1| = |A_1| = |B_2| = w$
        \item $S + B_1 \in \I_1$
        \item $S + B_1 - A_1 \in \I_2$
        \item $S - A_1 + B_2 \in \I_1$
        \item $S + B_2 \in \I_2$
    \end{enumerate}
\end{definition}

\begin{theorem}[from {\cite[Theorem 25]{chakrabarty2019faster}}] \label{afterupdateaugmentingset}
    Let $S \in \I_1 \cap \I_2$ be such that shortest augmenting paths in $G(S)$ have length $4$.
    Let $\Pi := (B_1, A_1, B_2)$ be an augmenting set in the exchange graph $G(S)$.
    Then, the set $S' := S \oplus \Pi := S + B_1 - A_1 + B_2$ is a common independent set.
\end{theorem}


Here, we recall the definition and property of {\em maximal augmenting sets}, which correspond to a maximal collection of shortest augmenting paths such that augmentation along them must increase the $(s, t)$-distance.
In the $(1 - \eps)$-approximation algorithms of Chakrabarty et al.~and Blikstad, they repeatedly find a maximal augmenting set and augment along it.

\begin{definition}[Maximal Augmenting Sets from {\cite[Definitions 31 and 35]{chakrabarty2019faster}}]
    Let $S \in \I_1 \cap \I_2$ be such that shortest augmenting paths in $G(S)$ have length $4$.
    Let $\Pi := (B_1, A_1, B_2)$ and $\tilde{\Pi} := (\tilde{B_1}, \tilde{A_1}, \tilde{B_2})$ be two augmenting sets.
    We say $\tilde{\Pi}$ \emph{contains} $\Pi$ if $B_1 \subseteq \tilde{B_1}$, $A_1 \subseteq \tilde{A_1}$ and $B_2 \subseteq \tilde{B_2}$.
    We use the notation $\Pi \subseteq \tilde{\Pi}$ to denote this.
    An augmenting set $\Pi$ is called \emph{maximal} if there exists no other augmenting set $\tilde{\Pi}$ containing $\Pi$.
\end{definition}

\begin{theorem}[from {\cite[Theorem 36]{chakrabarty2019faster}}] \label{thm:no_path}
    Let $S \in \I_1 \cap \I_2$ be such that shortest augmenting paths in $G(S)$ have length $4$.
    Let $\Pi$ be a maximal augmenting set in the exchange graph $G(S)$.
    Then, there is no augmenting path of length at most $4$ in $G(S \oplus \Pi)$.
\end{theorem}

Here, we recall the definition of {\em partial augmenting sets}, which are the relaxed form of augmenting sets.
In the algorithms of Chakrabarty et al.~and Blikstad for finding a maximal augmenting set, they keep track of the partial augmenting set, which is iteratively made close to a maximal augmenting set through some refine procedures.

\begin{definition}[Partial Augmenting Sets from {\cite[Definition 37]{chakrabarty2019faster}}]
    Let $S \in \I_1 \cap \I_2$ be such that shortest augmenting paths in $G(S)$ have length $4$.
    We say that $\Phi := (B_1, A_1, B_2)$ forms a \emph{partial augmenting set} in $G(S)$ if it satisfies the conditions (a), (c), (e), (f) of an \emph{augmenting set}, and the following two relaxed conditions:

    \begin{enumerate}
        \item[(b)] $|B_1| \geq |A_1| \geq |B_2|$
        \item[(d)] ${\rm rank}_2(S + B_1 - A_1) = {\rm rank}_2(S)$ 
    \end{enumerate}
\end{definition}

Chakrabarty et al.~presented an efficient algorithm to convert any partial augmenting set $\Phi$ into an augmenting set $\Pi$.

\begin{lemma}[from {\cite[Lemma 38]{chakrabarty2019faster}}] \label{lem:extract_augmenting_set}
    Given a partial augmenting set $\Phi = (B_1, A_1, B_2)$, using $O(n)$ independence oracle queries, we can find an augmenting set $\Pi = (B'_1, A'_1, B'_2)$ such that $B'_1 \subseteq B_1$, $A'_1 \subseteq A_1$ and $B'_2 = B_2$.
\end{lemma}

Chakrabarty et al.~\cite{chakrabarty2019faster} presented an algorithm to find a maximal augmenting set in the exchange graph $G(S)$, where shortest augmenting paths have length $2(\ell + 1)$, using $O(n^{1.5} \sqrt{\ell \log r})$ independence oracle queries; see \cite[Proof of Theorem 48 and Proof of Theorem 21]{chakrabarty2019faster}.
Blikstad \cite{blikstad2021breaking} improved the query complexity and presented an algorithm to find a maximal augmenting set using $O(n \sqrt{r \log r})$ independence oracle queries; see \cite[Lemma 35 and Proof of Theorem 1]{blikstad2021breaking}.

In Blikstad's algorithm for finding a maximal augmenting set, we keep track of a partial augmenting set and iteratively update it to closer to a maximal augmenting set.
To obtain a fast algorithm, Blikstad combines two algorithms \texttt{\textup{Refine}} \cite[Algorithm 4]{blikstad2021breaking} and \texttt{\textup{RefinePath}} \cite[Algorithm 5]{blikstad2021breaking}; see \cite[Algorithm 6 and Lemma 35]{blikstad2021breaking}.
Let $p \in [1, r]$ be the parameter that controls the trade-off of the two algorithms. 
He first applies \texttt{Refine} to find a partial augmenting set that is close enough to a maximal augmenting set, which uses $O(n r / p)$ independence oracle queries.
Then, he applies \texttt{\textup{RefinePath}} until the partial augmenting set becomes a maximal augmenting set, which uses $O(n p \log r)$ independence oracle queries.

In our algorithm, we apply only the first part \texttt{Refine} of Blikstad's algorithm, as formally stated as follows.

\begin{lemma}[follows from {\cite[Lemma 35]{blikstad2021breaking}}] \label{lem:blikstad_key_lemma}
    For any $p \in [1, r]$, given a common independent set $S \in \I_1 \cap \I_2$ and the distance layers $D_1, D_2,$ and $D_3$, using $O(\frac{n r}{p})$ independence oracle queries,
    we can find a partial augmenting set $\Phi = (B_1, A_1, B_2)$ such that the following properties hold.

    \begin{enumerate}[(i)]
        \item There is an augmenting set $\Pi = (B'_1, A'_1, B'_2)$ such that $B'_1 \subseteq B_1$, $A'_1 \subseteq A_1$ and $B'_2 = B_2$.
        \item We have $|B_1| - |B_2| \leq p$.
        \item There is a maximal augmenting set $\tilde{\Pi}$ of width at most $|B_1|$ in $G(S)$.  
    \end{enumerate}
\end{lemma}

To obtain Lemma \ref{lem:blikstad_key_lemma}, we apply \texttt{Refine} \cite[Algorithm 4]{blikstad2021breaking} (see also Algorithm \ref{alg:blikstadrefine}) until $|B_1| - |B_2| \leq p$, but at least once.
The procedure \texttt{Refine} makes the partial augmenting set $\Phi = (B_1, A_1, B_2)$ close to a maximal augmenting set.
To become $|B_1| - |B_2| \leq p$, we need to apply \texttt{\textup{Refine}} $O(|S|/p + 1)$ times; see \cite[Proof of Lemma 35]{blikstad2021breaking}.
Since each call of $\texttt{Refine}$ uses $O(n)$ independence oracle queries (see \cite[Lemma 29]{blikstad2021breaking}), the total number of independence oracle queries is $O(nr / p)$.

\begin{remark}
    Lemma \ref{lem:blikstad_key_lemma} is not explicitly stated in Blikstad's \cite{blikstad2021breaking} paper.
    By the property of a partial augmenting set and the argument in \cite[Proof of Lemma 35]{blikstad2021breaking}, it is clear that the conditions $(i)$ and $(ii)$ are satisfied.
    In particular, here we show that the condition $(iii)$ is also satisfied.
    In the algorithm in \cite[Algorithm 6 and Lemma 35]{blikstad2021breaking} for finding a maximal augmenting set, after applying \texttt{\textup{Refine}} until $|B_1| - |B_2| \leq p$, we apply \texttt{\textup{RefinePath}} \cite[Algorithm 5]{blikstad2021breaking} until the partial augmenting set becomes a maximal augmenting set.
    While \texttt{\textup{RefinePath}} modifies $B_1$, it does not increase $|B_1|$; see \cite[the second paragraph of the proof of Lemma 35]{blikstad2021breaking}.
    Thus, the partial augmenting set after applying \texttt{\textup{RefinePath}} is a maximal augmenting set of width at most $|B_1|$.
    Here, let this maximal augmenting set be $\tilde{\Pi}$.
\end{remark}

%% file: approx_algo.tex
\section{$(2/3-\varepsilon)$-Approximation Algorithm using Nearly-Linear Independence-Oracle Queries}

In this section, by providing a nearly-linear-independence-query $(2/3-\varepsilon)$-approximation algorithm for the matroid intersection problem, we prove Theorem \ref{thm:main_matroid_intersection_ind}, which we restate here.

\indalgo*

We use the following lemma to prove Theorem \ref{thm:main_matroid_intersection_ind}.

\begin{lemma} \label{lem:twothirdspath}
    Let $S \in \I_1 \cap \I_2$ be a common independent set such that shortest augmenting paths in $G(S)$ have length at least $6$.
    Then, $S$ is a $2/3$-approximate solution of the matroid intersection problem.
\end{lemma}

\begin{proof}
    We argue by contraposition.
    Let $S \in \I_1 \cap \I_2$ be a common independent set with $|S| < \frac{2}{3}r$.
    Then, we have $\frac{2 |S|}{r - |S|} + 2 < 6$.
    Thus, by Lemma \ref{lem:total_bound_matroid_intersection}, there exists an augmenting path of length at most $\frac{2 |S|}{r - |S|} + 2 < 6$ in $G(S)$.
    This completes the proof.
\end{proof}


    

Now, we give a proof of Theorem \ref{thm:main_matroid_intersection_ind}.
See Algorithm \ref{alg:total_algo} for the pseudocode of our algorithm.

\begin{proof}[Proof of Theorem \ref{thm:main_matroid_intersection_ind}]
    We give an algorithm to find a $(2/3 - \eps)$-approximate solution for the matroid intersection problem.
    In our algorithm, we first compute a maximal common independent set $S \in \I_1 \cap \I_2$, which uses $O(n)$ independence oracle queries.
    Here, let $\bar{r}$ be the size of this maximal common independent set $S$.
    We note that $r \leq 2 \bar{r}$ since any maximal common independent set is at least half the size of a maximum common independent set (see e.g., \cite[Proposition 13.26]{korte2011combinatorial}).

    Then, we compute the distance layers $D_1$, $D_2$, and $D_3$ in the exchange graph $G(S)$, which, by Lemma \ref{lem:find_dist_layers_ind}, uses $O(n + r \log r)$ independence oracle queries.
    Here, there is no augmenting path of length $2$ in $G(S)$, since $S$ is a maximal common independent set.
    Thus, shortest augmenting paths in $G(S)$ have length at least $4$.
    
    Next, by checking whether $S + v \in \I_2$ for each $v \in D_3$,  we check whether the $(s, t)$-distance is $4$ or not, which uses $O(n)$ independence oracle queries.
    If the $(s, t)$-distance is more than $4$, then $S$ is already a $2/3$-approximate solution by Lemma \ref{lem:twothirdspath}, so we output $S$.
    Otherwise, we apply Lemma \ref{lem:blikstad_key_lemma} in which we set the parameter $p = \eps \bar{r}$, which uses $O(n / \eps)$ independence oracle queries.\footnote{We apply \texttt{\textup{Refine}} (Algorithm \ref{alg:blikstadrefine}) only $O(1/\eps)$ times.}
    Then, we obtain a partial augmenting set $\Phi = (B_1, A_1, B_2)$.
    By applying Lemma \ref{lem:extract_augmenting_set} for the obtained partial augmenting set $\Phi$, using $O(n)$ independence oracle queries, we can find an augmenting set $\Pi = (B'_1, A'_1, B'_2)$ such that $B'_1 \subseteq B_1$, $A'_1 \subseteq A_1$ and $B'_2 = B_2$.
    Then, we output the set $S \oplus \Pi$.

    By Theorem \ref{afterupdateaugmentingset}, the set $S \oplus \Pi$ is a common independent set.
    Now, we show that the obtained solution $S \oplus \Pi$ is a $(2/3 - \eps)$-approximate solution.

    By the condition $(iii)$ in Lemma \ref{lem:blikstad_key_lemma}, there is a maximal augmenting set $\tilde{\Pi} = (\tilde{B_1}, \tilde{A_1}, \tilde{B_2})$ in $G(S)$ such that $|\tilde{B_1}| \leq |B_1|$.
    By Theorems \ref{afterupdateaugmentingset} and \ref{thm:no_path} and Lemma \ref{lem:twothirdspath}, $S \oplus \tilde{\Pi}$ is a $2/3$-approximate solution, so we have $|S \oplus \tilde{\Pi}| \geq \frac{2}{3} r$.
    Moreover, by the condition $(ii)$ in Lemma \ref{lem:blikstad_key_lemma}, we have $|\tilde{B_1}| \leq |B_1| \leq |B_2| + p = |B'_2| + p = |B'_1| + p$. 
    Thus, we have $|S \oplus \tilde{\Pi}| - |S \oplus \Pi| = |\tilde{B_1}| - |B'_1| \leq p = \eps \bar{r} \leq \eps r$.
    Therefore, we have $|S \oplus \Pi| \geq (2/3 - \eps)r$, which completes the proof.
\end{proof}

\begin{algorithm}[t]
    Compute a maximal common independent set $S \in \I_1 \cap \I_2$.  \tcp*[f]{Write $\bar{r} = |S|$.} \\
    Compute the distance layers $D_1, D_2$, and $D_3$ in the exchange graph $G(S)$. \label{line:our_algo_compute_distance} \\
    \If{the $(s, t)$-distance in $G(S)$ is more than $4$} {
        \Return $S$
    }
    Apply \texttt{Refine} (Algorithm \ref{alg:blikstadrefine}) until $|B_1| - |B_2| \leq \eps \bar{r}$. \\
    Apply Lemma \ref{lem:extract_augmenting_set} for the obtained partial augmenting set $\Phi = (B_1, A_1, B_2)$ and then obtain an augmenting set $\Pi = (B'_1, A'_1, B'_2)$ such that $B'_1 \subseteq B_1$, $A'_1 \subseteq A_1$ and $B'_2 = B_2$. \\
    \Return $S \oplus \Pi$
    \caption{$(2/3-\eps)$-approximation algorithm for the matroid intersection problem}\label{alg:total_algo} 
\end{algorithm}

%% file: approx_algo_using_rank_queries.tex
\section{$(2/3-\varepsilon)$-Approximation Algorithm using Linear Rank-Oracle Queries}

In this section, by providing a linear-rank-query $(2/3-\varepsilon)$-approximation algorithm for the matroid intersection problem, we prove Theorem \ref{thm:main_matroid_intersection_rank}.
We show that Algorithm \ref{alg:total_algo} can be implemented as a $(2/3 - \eps)$-approximation algorithm with linear rank oracle queries.
We restate Theorem \ref{thm:main_matroid_intersection_rank} here.

\rankalgo*

Chakrabarty--Lee--Sidford--Singla--Wong \cite{chakrabarty2019faster} showed that a $(1 - \eps)$-approximate solution can be obtained with $O(\frac{n}{\eps} \log n)$ rank oracle queries by using the binary search technique; see \cite[Theorem 17]{chakrabarty2019faster}.
In their $(1 - \eps)$-approximation algorithm, they find some edges using the binary search technique.
For each edge, this requires $O(\log n)$ rank oracle queries. 
On the other hand, we show that a $(2/3 - \eps)$-approximate solution can be obtained with only $O(n/\eps)$ rank oracle queries.
In our $(2/3 - \eps)$-approximation algorithm, we find an almost maximal augmenting set without needing to identify any specific edges.
Hence, the rank oracle query complexity of our algorithm is linear without any logarithmic factor.


We use the following lemma to prove Theorem \ref{thm:main_matroid_intersection_rank}.

\begin{lemma} \label{lem:find_dist_layers_rank}
    Given a common independent set $S \in \I_1 \cap \I_2$, using $O(n)$ rank oracle queries, we can find the distance layers $D_1 \subseteq V \setminus S, D_2 \subseteq S$, and $D_3 \subseteq V \setminus S$.
\end{lemma}

\begin{proof}
    Our algorithm for computing $D_1, D_2$, and $D_3$ is almost the same as Algorithm \ref{alg:GetDistance}.
    Since a single rank oracle query can determine whether a given subset is independent or not, Algorithm \ref{alg:GetDistance} except for computing $D_2$ can be implemented with $O(n)$ rank oracle queries.
    Now, we show how $D_2$ can be computed with $O(n)$ rank oracle queries.
    For each $u \in S$, we check whether $\textrm{rank}_{\M_2}(S + D_1 - u) \geq \textrm{rank}_{\M_2}(S)$ holds or not.
    If $\textrm{rank}_{\M_2}(S + D_1 - u) \geq \textrm{rank}_{\M_2}(S)$ holds, then there exists an element $v \in D_1$ such that $S + v - u \in \I_2$, and consequently $u \in D_2$.
    Otherwise, there does not exist an element $v \in D_1$ such that $S + v - u \in \I_2$, and consequently $u \notin D_2$.
    This completes the proof.
\end{proof}

\begin{proof}[Proof of Theorem \ref{thm:main_matroid_intersection_rank}]
    Now, we show how Algorithm \ref{alg:total_algo} can be implemented with $O(n / \eps)$ rank oracle queries.
    As mentioned in the proof of Theorem \ref{thm:main_matroid_intersection_ind}, all parts of Algorithm \ref{alg:total_algo}, except for Line \ref{line:our_algo_compute_distance}, can be implemented with $O(n / \eps)$ independence oracle queries.
    Since a single rank oracle query can determine whether a given subset is independent or not, we can implement all parts of Algorithm \ref{alg:total_algo}, except for Line \ref{line:our_algo_compute_distance}, with $O(n / \eps)$ rank oracle queries.
    In addition, by Lemma \ref{lem:find_dist_layers_rank}, we can implement Line \ref{line:our_algo_compute_distance} with $O(n)$ rank oracle queries, which completes the proof.
\end{proof}

%% file: streaming_algo.tex
\section{$(2/3-\varepsilon)$-Approximation Algorithm in the Semi-Streaming Model}

In this section, by providing a $(2/3-\varepsilon)$-approximation algorithm for the matroid intersection problem in the semi-streaming model, we prove Theorem \ref{thm:main_matroid_intersection_streaming}.
We show that Algorithm \ref{alg:total_algo} can also be implemented in the semi-streaming model as a $(2/3 - \eps)$-approximation algorithm using a constant number of passes.
We restate Theorem \ref{thm:main_matroid_intersection_streaming} here.

\streamingalgo*



Before we present how Algorithm \ref{alg:total_algo} can be implemented in the semi-streaming model, we describe the outline of Blikstad's \cite{blikstad2021breaking} algorithm of Lemma \ref{lem:blikstad_key_lemma}.
This is because some of his results will be used in our argument later.
In this paper, we skip the proof of the correctness of the algorithm; see \cite[Sections 3.1 and 3.3]{blikstad2021breaking} for the proof.
Recall that, in Lemma \ref{lem:blikstad_key_lemma}, we apply \texttt{Refine} only $O(1/\eps)$ times.
See Algorithm \ref{alg:blikstadrefine} for the pseudocode of \texttt{Refine}.

In Blikstad's algorithm, we maintain three types of elements in each layer (see \cite[Section 3.1]{blikstad2021breaking}): 
\begin{itemize}
    \item {\em fresh.} Denoted by $F_i \subseteq D_i$. These elements are candidates that could be added to the partial augmenting set.
    \item {\em selected.} Denoted by $B_1, A_1, B_2$. These elements form the current partial augmenting set $\Pi = (B_1, A_1, B_2)$.
    \item {\em removed.} Denoted by $R_i \subseteq D_i$. These elements are deemed useless, and then we can disregard them.
\end{itemize}
For convenience, we also define imaginary layers $D_0$ and $D_4$ with $A_0 = R_0 = F_0 = D_0 = A_2 = R_4 = F_4 = D_4 = \emptyset$.

In \texttt{Refine} (Algorithm \ref{alg:blikstadrefine}), we iteratively apply \texttt{RefineAB} (Algorithm \ref{alg:blikstadphaseab}) and \texttt{RefineBA} (Algorithm \ref{alg:blikstadphaseba}), and then update the types of elements.
Initially, we begin with all elements being {\em fresh}. 
Elements can change their types from {\em fresh} $\rightarrow$ {\em selected} $\rightarrow$ {\em removed}, but their types cannot be changed in the other direction.
Note that an element of type {\em fresh} can change to type {\em removed} without first going through type {\em selected}.
    
\begin{algorithm}[t]
    Find maximal $B \subseteq F_{2k+1}$ s.t. $S - A_k + B_{k + 1} + B \in \I_1$ \\
    $B_{k + 1} \gets B_{k + 1} + B$, $F_{2k+1} \gets F_{2k+1} - B$ \\
    Find maximal $A \subseteq A_k$ s.t. $S - A_k + B_{k + 1} + A \in \I_1$ \\
    $A_k \gets A_k - A$, $R_{2k} \gets R_{2k} + A$ \\
    \caption{\texttt{RefineAB}$(k)$ (from {\cite[Algorithm 1]{blikstad2021breaking}}); called \texttt{Refine1} in {\cite[Algorithm 9]{chakrabarty2019faster}}}\label{alg:blikstadphaseab}
\end{algorithm}

\begin{algorithm}[t]
    Find maximal $B \subseteq B_k$ s.t. $S - (D_{2k} - R_{2k}) + B \in \I_2$ \\
    $R_{2k-1} \gets R_{2k-1} + B_k \setminus B$, $B_{k} \gets B$ \\
    Find maximal $A \subseteq F_{2k}$ s.t. $S - (D_{2k} - R_{2k}) + B_{k} + A \in \I_2$ \\
    $A_k \gets A_k + F_{2k} \setminus A$, $F_{2k} \gets A$ \\
    \caption{\texttt{RefineBA}$(k)$ (from {\cite[Algorithm 2]{blikstad2021breaking}}); called \texttt{Refine2} in {\cite[Algorithm 10]{chakrabarty2019faster}}}\label{alg:blikstadphaseba}
\end{algorithm}

\begin{algorithm}[t]
    \For{$k = 1, 0$} {
        \texttt{RefineBA}$(k+1)$ \\
        \For{$x \in F_{2k+1}$} { \label{line:refine1}
            \If{$S - A_k + B_{k + 1} + x \in \I_1$} {
                \If{$S - A_{k + 1} - F_{2k+2} + B_{k + 1} + x \in \I_2$} {
                    $B_{k + 1} \gets B_{k + 1} + x$, $F_{2k+1} \gets F_{2k+1} - x$
                } \Else {
                    $R_{2k+1} \gets R_{2k+1} + x$, $F_{2k+1} \gets F_{2k+1} - x$
                }
            }
        } \label{line:refine2}
        \texttt{RefineBA}$(k+1)$ \\
        \texttt{RefineAB}$(k)$ \\
    }
    \caption{\texttt{Refine} (from {\cite[Algorithms 3 and 4]{blikstad2021breaking}})}\label{alg:blikstadrefine}
\end{algorithm}


Now, by providing how Algorithm \ref{alg:total_algo} can be implemented in the semi-streaming model, we give a proof of Theorem \ref{thm:main_matroid_intersection_streaming}.

\begin{proof}[Proof of Theorem \ref{thm:main_matroid_intersection_streaming}]
    We show how Algorithm \ref{alg:total_algo} can be implemented using $O(\frac{r_1 + r_2}{\eps})$ memory and $O(\frac{1}{\eps})$ passes in the semi-streaming model.

    In our algorithm, We first compute a maximal common independent set $S \in \I_1 \cap \I_2$ using one pass of the stream.
    Note that we can easily compute a maximal set in a single pass.
    We store the set $S$ explicitly using $O(r)$ memory.

    Then, we compute $D_2$ in the exchange graph $G(S)$ using one pass of the stream.
    For each $v \in V \setminus S$, if $S + v \in \I_1$, then we find all elements $u \in S$ such that $S + v - u \in \I_2$ and add them to $D_2$.

    We store $D_2$ explicitly using $O(r)$ memory. In our algorithm, we store the types of elements (i.e., $F_2, A_1, R_2$) in $D_2$ explicitly using $O(r)$ memory.
    Whenever an element $v \in V \setminus S$ arrives in the stream, we can determine whether $v \in D_1$, $v \in D_3$, or otherwise, by checking whether $S + v \in \I_1$ and whether there exists an element $u \in D_2$ such that $S + v - u \in \I_1$.
    Thus, we can maintain the distance layers $D_1$ and $D_3$ implicitly.

    Next, we apply \texttt{Refine} (Algorithm \ref{alg:blikstadrefine}) $O(1/\eps)$ times.
    In our implementation of \texttt{Refine} in the semi-streaming model, we replace Lines \ref{line:refine1}--\ref{line:refine2} in \texttt{Refine} with \texttt{UpdateABA}$(k)$ (Algorithm \ref{alg:blikstadrefine_streaming}).
    By Claim \ref{claim:how_to_implement_in_streaming_setting}, our implementation of \texttt{Refine} can also correctly find a desired partial augmenting set.

    \begin{claim} \label{claim:how_to_implement_in_streaming_setting}
        Even if we replace Lines \ref{line:refine1}--\ref{line:refine2} in \texttt{Refine} with \texttt{UpdateABA}$(k)$ (Algorithm \ref{alg:blikstadrefine_streaming}), the procedure \texttt{\textup{Refine}} can correctly find a partial augmenting set that satisfies the conditions $(i)$--$(iii)$ in Lemma \ref{lem:blikstad_key_lemma}.
    \end{claim}

    \begin{proof}
        Consider the order of the elements in $F_{2k+1}$ satisfying the following condition: an element added to $B_{k + 1}$ in the first for loop of \texttt{UpdateABA}$(k)$ appears before any element that is not added to $B_{k + 1}$ in that loop.
        Suppose that, in the execution of Lines {\ref{line:refine1}--\ref{line:refine2}} in \texttt{Refine}, the elements in $F_{2k+1}$ arrive in this order.
        Then, the changes of types of elements in $D_{2k+1}$ by \texttt{UpdateABA}$(k)$ call are exactly same as by the execution of Lines {\ref{line:refine1}--\ref{line:refine2}} in \texttt{Refine} (Algorithm \ref{alg:blikstadrefine}).
        Regardless of the order of the elements in the for loop in Line \ref{line:refine1}, the implementation of \texttt{\textup{Refine}} in Algorithm \ref{alg:blikstadrefine} can find a desired partial augmenting set.  
        Therefore, \texttt{Refine} with Lines \ref{line:refine1}--\ref{line:refine2} replaced by \texttt{UpdateABA}$(k)$ correctly finds a desired partial augmenting set, which completes the proof.
    \end{proof}

    \begin{algorithm}[t]
        \For(\tcp*[f]{first pass}){$x \in F_{2k+1}$} { 
            \If{$S - A_k + B_{k + 1} + x \in \I_1$ and $S - A_{k + 1} - F_{2k+2} + B_{k + 1} + x \in \I_2$} {
                $B_{k + 1} \gets B_{k + 1} + x$, $F_{2k+1} \gets F_{2k+1} - x$
            }
        } 
        Store the current $A_k$ and $B_{k + 1}$ as $A_k^{(i)}$ and $B_{k + 1}^{(i)}$, respectively. \\
        \tcp{Assume that this procedure is the $i$-th call of \texttt{UpdateABA}$(k)$ in the entire matroid intersection algorithm.}
        \For(\tcp*[f]{second pass}){$x \in F_{2k+1}$} {
            \If{$S - A_k + B_{k + 1} + x \in \I_1$} {
                $R_{2k+1} \gets R_{2k+1} + x$, $F_{2k+1} \gets F_{2k+1} - x$ \label{line:streaming1}
            }
        } 
        \caption{\texttt{UpdateABA}$(k)$ (implementation of Lines {\ref{line:refine1}--\ref{line:refine2}} in \texttt{Refine} (Algorithm \ref{alg:blikstadrefine}) in the semi-streaming model)} \label{alg:blikstadrefine_streaming}
    \end{algorithm}

    Now, we show that each call of the modified \texttt{Refine} can be implemented with $O(1)$ passes.
    Since we can compute a maximal set in a single pass, both \texttt{RefineAB} (Algorithm \ref{alg:blikstadphaseab}) and \texttt{RefineBA} (Algorithm \ref{alg:blikstadphaseba}) can be implemented with $2$ passes. 
    In addition, \texttt{UpdateABA} (Algorithm \ref{alg:blikstadrefine_streaming}) can also be implemented with $2$ passes.
    
    Next, we show that we can maintain the types of elements (i.e., $F_1, B_1, R_1, F_3, B_2, R_3$) in $D_1$ and $D_3$ implicitly. 
    To do this, we explicitly maintain the following:

    \begin{itemize}
        \item We store the current partial augmenting set $(B_1, A_1, B_2)$. Since the number of selected elements is always at most $r_1 + r_2$, we use $O(r_1 + r_2)$ memory to store it.
        
        \item We store all removed elements whose type has changed from {\em selected} to {\em removed}. 
        Only \texttt{RefineAB} and \texttt{RefineBA} change the types of elements from {\em selected} to {\em removed}. 
        In each call of \texttt{RefineAB} and \texttt{RefineBA}, this change occurs only for at most $r_1 + r_2$ elements.
        Since the number of \texttt{RefineAB} and \texttt{RefineBA} calls is $O(1/\eps)$, we use $O((r_1 + r_2)/\eps)$ memory to store them in the entire algorithm.

        \item We store $A_k^{(i)}$ and $B_{k + 1}^{(i)}$ just after the first for loop in \texttt{UpdateABA} (Algorithm \ref{alg:blikstadrefine_streaming}). 
        Since the number of \texttt{UpdateABA} calls is $O(1/\eps)$, we use $O((r_1 + r_2)/\eps)$ memory to store them in the entire algorithm.
    \end{itemize}

    Here, we note that, for an element whose current type is {\em removed}, there are only the following two cases.
    \begin{itemize}
        \item The type has changed from {\em selected} to {\em removed} by \texttt{RefineAB} or \texttt{RefineBA}.
        \item The type has changed from {\em fresh} to {\em removed} by \texttt{UpdateABA}.
    \end{itemize}
    
    Whenever an element $v \in D_1 \cup D_3$ arrives, we can identify the current type of $v$ in the following way:

    \begin{itemize}
        \item If the current type of $v$ is {\em selected}, then we can easily identify it.
        \item If the current type of $v$ is {\em removed}, then we identify it in the following way:
        \begin{itemize}
            \item In the case where the type of $v$ has changed from {\em selected} to {\em removed} by \texttt{RefineAB} or \texttt{RefineBA}, we can conclude that the current type of $v$ is {\em removed}.

            \item In the case where the type of $v$ has changed from {\em fresh} to {\em removed} by Line \ref{line:streaming1} in \texttt{UpdateABA} (Algorithm \ref{alg:blikstadrefine_streaming}), we identify it by simulating all previous executions of the second for loop in \texttt{UpdateABA}.
            More precisely, let $C$ be the number of \texttt{UpdateABA} called so far.
            If $v \in D_{2 k + 1}$ and there is an index $i \leq C$ such that there are $A_k^{(i)}$ and $B_{k + 1}^{(i)}$ such that $S - A_k^{(i)} + B_{k + 1}^{(i)} + v \in \I_1$, then we conclude that the current type of $v$ is {\em removed}.
        \end{itemize}

        \item Otherwise, the current type of $v$ is {\em fresh}.
    \end{itemize}
    Therefore, each call of the modified \texttt{Refine} can be implemented with $O(1)$ passes.

    After applying \texttt{Refine} $O(1/\eps)$ times, we obtain a partial augmenting set $\Phi = (B_1, A_1, B_2)$ such that $|B_1| - |B_2| \leq \eps \bar{r}$. 
    Then, by applying Lemma \ref{lem:extract_augmenting_set} for the obtained partial augmenting set $\Phi = (B_1, A_1, B_2)$, we obtain an augmenting set $\Pi = (B'_1, A'_1, B'_2)$ such that $B'_1 \subseteq B_1$, $A'_1 \subseteq A_1$ and $B_2 = B'_2$.
    The algorithm in Lemma \ref{lem:extract_augmenting_set} can be implemented without an additional pass of the stream, because it only accesses the set $S$ and the partial augmenting set $\Phi = (B_1, A_1, B_2)$; see \cite[Proof of Lemma 38]{chakrabarty2019faster}.
    
    Finally, we output the set $S \oplus \Pi$.
    By the same argument as Proof of Theorem \ref{thm:main_matroid_intersection_ind}, the set $S \oplus \Pi$ is a $(2/3-\eps)$-approximate solution, which completes the proof.
\end{proof}

%% file: conclusion.tex
\section{Concluding Remarks} \label{sec:concluding_remarks}

We have observed that a $(2/3 - \eps)$-approximate solution can be obtained by terminating Blikstad's \cite{blikstad2021breaking} algorithm early.
Then, we obtained a deterministic nearly-linear-independence-query $(2/3-\eps)$-approximation algorithm for the matroid intersection problem.

We were also able to use this observation in the streaming model of computation.
Then, we obtain a $(2/3 - \eps)$-approximation constant-pass semi-streaming algorithm for the matroid intersection problem.
This result is a generalization of the $(2/3 - \eps)$-approximation bipartite matching semi-streaming algorithm of Feigenbaum--Kannan--McGregor--Suri--Zhang \cite{feigenbaum2005graph}, who initiated the study of matching algorithms in the semi-streaming model.
Soon after, McGregor \cite{mcgregor2005finding} first obtained a $(1-\eps)$-approximation constant-pass semi-streaming algorithm for the maximum matching problem.
Since then, there have been many studies for $(1-\eps)$-approximation constant-pass semi-streaming algorithms for the maximum matching problem in the literature \cite{eggert2009bipartite,ahn2011laminar,eggert2012bipartite,kapralov2013better,tirodkar2018deterministic,gamlath2019weighted,assadi2021auction,fischer2022deterministic,huang20231}.
Then, it is natural to ask whether we can obtain a $(1-\eps)$-approximation constant-pass semi-streaming algorithm for the matroid intersection problem.

In the current Blikstad's algorithm for finding an augmenting set, it is necessary to perfectly find an augmenting set of length $2k$ before finding an augmenting set of length $2k + 2$.
By finding a maximal common independent set, we can perfectly find an augmenting set of length 2. This means that we can start from a state where shortest augmenting paths have length 4 or greater.
In our current algorithm, since we do not perfectly find an augmenting set of length 4, we are unable to find an augmenting set of length 6 or greater.
Thus, our current algorithm achieves only a $(2/3 - \eps)$-approximation rather than a $(1 - 1/k - \eps)$-approximation.
However, we believe our study could be an important step to obtain a deterministic nearly-linear-independence-query $(1-\eps)$-approximation algorithm.